\documentclass[conference]{IEEEtran}
\IEEEoverridecommandlockouts

\usepackage{cite}
\usepackage{amsmath,amssymb,amsfonts}
\usepackage{amsthm}
\usepackage{graphicx}
\usepackage{textcomp}
\usepackage{comment}
\usepackage{xcolor}
\usepackage{algorithm}       
\usepackage{algpseudocode}   
\usepackage[table]{xcolor}
\usepackage{tcolorbox}
\newtheorem{definition}{Definition}
\newtheorem{corollary}{Corollary}
\usepackage[numbers]{natbib}
\newtheorem{theorem}{Theorem}

\usepackage{filecontents}
\usepackage{tikz}
\usepackage{graphicx}

\usepackage{threeparttable}
\usepackage[hidelinks,breaklinks=true]{hyperref}
\usepackage{pifont}
\usepackage{color}
\usepackage{caption}
\usepackage{subcaption}
\usepackage{enumitem}
\usepackage{cleveref}
\usepackage{bookmark}
\usepackage{multirow} 
\usepackage{booktabs}
\usepackage{siunitx}
\usepackage{tabularx}
\newcommand{\cmark}{\ding{51}}
\newcommand{\xmark}{\ding{55}}

\renewcommand{\S}{Sec.}
\newcommand{\F}{Fig.}
\newcommand{\T}{Table }
\newcommand{\E}{Eq.}

\newcommand{\Alg}{Algorithm }

\algrenewcommand{\algorithmiccomment}[1]{%
  \hfill{\color{green!50!black}$\triangleright$ #1}%
}

\def\BibTeX{{\rm B\kern-.05em{\sc i\kern-.025em b}\kern-.08em
    T\kern-.1667em\lower.7ex\hbox{E}\kern-.125emX}}
    
\begin{document}

\title{CertMoE: Robustness Certification for Mixture-of-Experts LLMs via Routing-Aware Smoothing}

\author{
\IEEEauthorblockN{
Qingyue Wang\IEEEauthorrefmark{1},
Yudong Gao\IEEEauthorrefmark{1},
Lixia Wang\IEEEauthorrefmark{2},
Ruixuan Huang\IEEEauthorrefmark{1},
Yuanyuan Yuan\IEEEauthorrefmark{3},
Yangqiu Song\IEEEauthorrefmark{1},
Shuai Wang\IEEEauthorrefmark{1}
}
\IEEEauthorblockA{\small
\IEEEauthorrefmark{1}The Hong Kong University of Science and Technology\\
\IEEEauthorrefmark{2}
Academy of Mathematics and Systems Science,
Chinese Academy of Sciences\\
\IEEEauthorrefmark{3}Tsinghua University
}
\thanks{Corresponding author: Yuanyuan Yuan
(yyyuan@mail.tsinghua.edu.cn).
Other authors' emails:
Qingyue Wang (qingyue.wang2018@gmail.com);
Yudong Gao (ygaodj@connect.ust.hk);
Lixia Wang (wanglixia@amss.ac.cn);
Ruixuan Huang (rhuangbi@connect.ust.hk);
Yangqiu Song (yqsong@cse.ust.hk);
Shuai Wang (shuaiw@cse.ust.hk).}
}

\maketitle

\begin{abstract}
Provable robustness is a cornerstone of trustworthy machine learning, aiming to guarantee model predictions under adversarial perturbations. Existing certification frameworks, inspired by abstract interpretation or symbolic execution, typically assume a fixed computation graph, where input perturbations or variations affect only model outputs but not the underlying execution path. This assumption, however, is violated in Mixture-of-Experts (MoE) Large Language Models (LLMs), whose input-dependent routing induces discrete and dynamic execution paths. Crucially, this dynamic computation changes the certified object from a single fixed function to a composition of input-dependent sub-functions, leading to unreliable certification guarantees.

This paper presents \textsc{CertMoE}, the first routing-aware robustness certification framework for MoE LLMs, which explicitly accounts for input-induced routing dynamics in certification. Specifically, we introduce a relaxed routing-containment criterion, which certifies that perturbed Top-$K$ routing is contained within a larger certified expert pool. Based on this guarantee, \textsc{CertMoE} then constructs the corresponding induced sub-network and certifies its prediction consistency. This hierarchical certification procedure yields a joint probabilistic guarantee over routing containment and sub-network robustness, providing sufficient conditions under which the original prediction remains invariant.
Extensive experiments across recent MoE LLMs demonstrate that dynamic routing makes certificates issued by existing methods frequently violated, while \textsc{CertMoE} reduces the empirical violation rate with routing-aware guarantees, e.g., from more than $30\%$ to $0$, yielding certificates that more faithfully reflect the robustness of MoE models.
\end{abstract}


\section{Introduction}


Mixture-of-Experts (MoE)~\cite{jacobs1991adaptive, shazeer2017outrageously} has become the key architecture for scaling large language models (LLMs) while keeping inference costs manageable. It has been adopted in the backbone of general-purpose systems such as Gemini~\cite{team2023gemini}, Deepseek~\cite{dai2024deepseekmoe}, and Grok~\cite{xai_grok41_2025}. Additionally, recent work has introduced MoE models into software engineering tasks, such as vulnerability detection~\cite{yang2025one} and repair~\cite{yuan2025leveraging}, where model errors can directly affect system security.
As MoE LLMs are increasingly deployed in reliability- and security-sensitive settings~\cite{amershi2019software}, it becomes pressing to establish robustness guarantees for MoE LLMs in a meaningful and reliable way.

A principled way to establish such guarantees is through formal robustness certification~\cite{gowal2018effectiveness,cohen2019certified}, which provides provable robustness guarantees by verifying prediction invariance within a specified input perturbation or variation region. Existing certification techniques, including deterministic methods that are based on interval bound propagation~\cite{mao2023understanding}, abstract interpretation~\cite{Gehr2018AI2SA}, symbolic execution~\cite{symbolexecu18}, etc., and probabilistic methods like randomized smoothing~\cite{cohen2019certified,mangal2019robustness}, typically rely on the assumption that the computation graph of the certified model is fixed. Thus, the model can be represented as a single, unchanged function throughout the certification procedure, enabling output behaviors to be bounded or probabilistically controlled over the specified input region.



However, this assumption is violated by MoE LLMs, whose input-dependent routing decisions induce dynamic computation graphs (see \F~\ref{fig:motivate_arch}). In particular, inputs within the certified region can alter the activated expert set, causing the model to execute through different computation paths, and consequently leading to different underlying functions for the same certified input region. As illustrated in~\F~\ref{fig:motivate_cert}, a minor perturbation can switch the activated experts from $\{E_1, E_4\}$ to $\{E_2, E_4\}$, turning the certified model from $f'$ to $f''$. 
Consequently, although existing methods may still be applied without crashes, \textit{the issued certificate becomes arbitrary because the certified model has already changed.}
Empirically, on the SST2 dataset (see \S~\ref{sec:main_results}), we find that $>30\%$ of certificates issued by existing methods are invalid: standard attacks can frequently flip the model's outputs within the input regions where the model is certified to be robust.

\captionsetup[subfigure]{labelformat=parens,labelsep=none,font=normalsize}

\begin{figure*}[t]
\centering
\subcaptionbox{~Architecture of Dense and MoE Models. \label{fig:motivate_arch}}{%
    \includegraphics[scale=0.36]{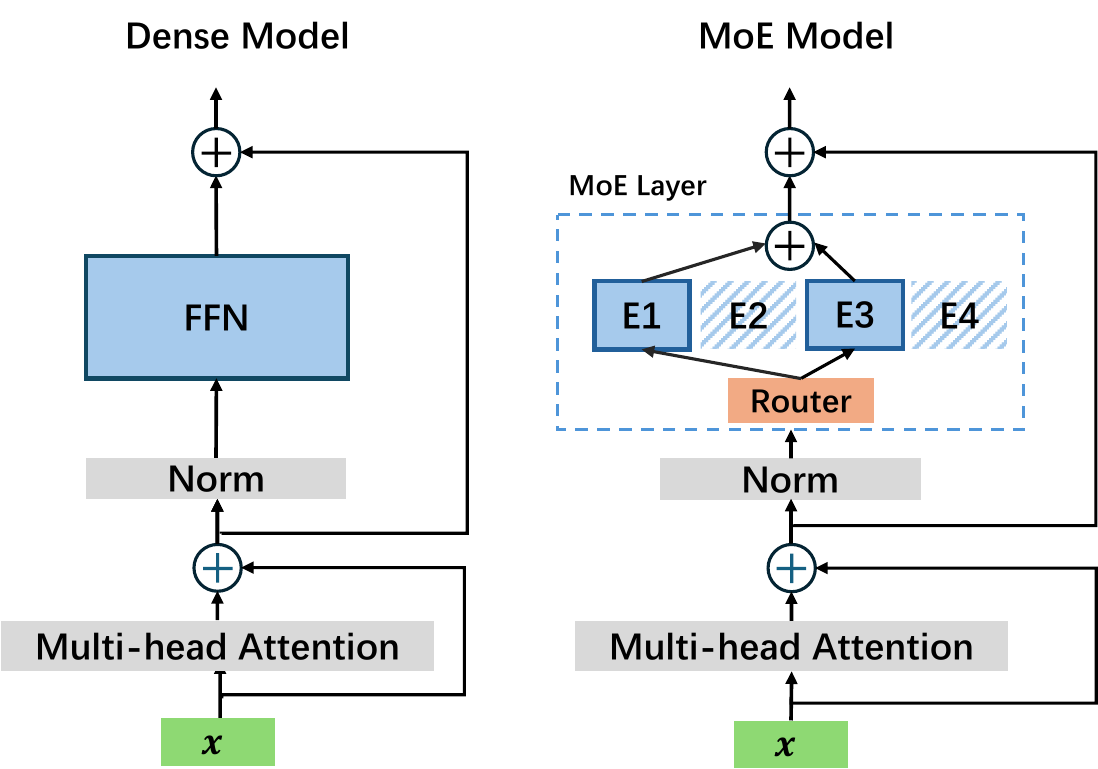}
}%
\hspace{3em}%
\subcaptionbox{~Certification on Dense and MoE Models. \label{fig:motivate_cert}\phantom{}}{%
    \includegraphics[scale=0.36]{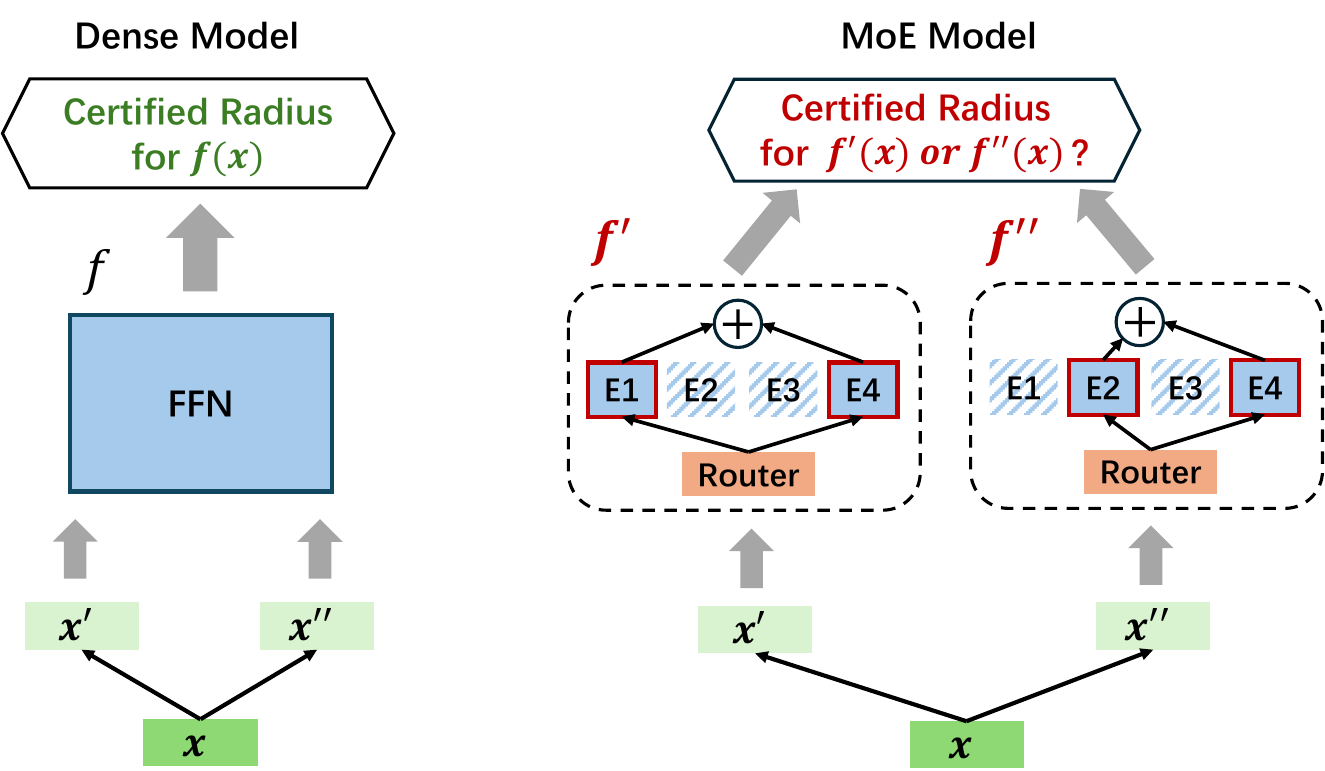}
}%

\caption{
Comparison between dense and MoE models.
(a) Architecture: MoE replaces the FFN layer with routed experts.
(b) Certification: In dense models, the model function $f$ remains static under input perturbations ($x', x''$). In MoE models, input perturbations trigger different routing paths, leading to functional instability ($f' \neq f''$) that complicates the derivation of a consistent certified radius.
Only a single MoE layer is shown for simplicity.
}
\label{fig:motivate}
\end{figure*}

Existing studies have begun to explore and understand the adversarial robustness of MoE models. Prior work shows that routing structure, expert similarity, and expert redundancy can substantially affect the robustness of sparse MoEs~\cite{puigcerver2022adversarial}. More recent work further studies robustness from the perspective of output sensitivity and robust training~\cite{zhangoptimizing}. These studies provide valuable insights into the robustness behavior of MoE models and the architectural factors that influence it. However, they mainly focus on analyzing robustness factors or improving training-time robustness, while how to perform  certification for the MoE model remains less clear.

This paper points out that robustness certification for MoE models should explicitly account for perturbation-induced changes in model structure: a certificate is meaningful only if the model's execution paths are covered w.r.t. the certified input region. To this end, we propose \textsc{CertMoE}, the first \emph{routing-aware} robustness certification framework for MoE LLMs. \textsc{CertMoE} extends existing certification methods to MoE models by explicitly accounting for perturbation-induced computation graph changes. In this paper, we instantiate \textsc{CertMoE} with randomized smoothing, a certification framework that scales to modern LLMs. Specifically, \textsc{CertMoE} decomposes certification into a hierarchical procedure that certifies both routing containment and prediction stability. We first aggregate token-level expert selections into sequence-level routing and certify their $\rho$-relaxed robustness under bounded input perturbations, where $\rho$ decides the size of the certified expert pool.
This avoids overly conservative exact-routing requirements, while still covering the perturbation-induced execution paths. After that, we restrict the candidate experts to the certified routing pools to construct an induced sub-network, and certify the robustness of this sub-network's output. Finally, a union bound combines the routing and output certificates into an end-to-end robustness guarantee.


We develop \textsc{CertMoE} on top of two smoothing-based certification frameworks, RanMASK~\cite{zeng2023certified} and SelfDenoise~\cite{ji2024advancing}, and conduct experiments on two classification tasks across three representative MoE LLMs.
Both baselines and their \textsc{CertMoE}-enhanced variants are evaluated using \emph{violation rate} (VR), which measures the fraction of issued certificates for which standard attacks find adversarial examples within the certified input regions.
We find that certificates from existing methods are violated frequently, with VR exceeding $30\%$ on SST-2. By incorporating routing-awareness, \textsc{CertMoE} substantially reduces such violations.
Further ablations show that the certified expert-pool size of routing-containment is critical: a larger pool relaxes the routing-containment condition and makes \textsc{CertMoE} more likely to assign certificates to inputs, but overly permissive pools weaken the routing-containment constraint and may lead to more certificate violations. Beyond final certificates, \textsc{CertMoE} can decompose certification outcomes into routing-containment and sub-network components, enabling failure attribution. As an on-demand certification procedure, it incurs roughly 2$\times$ runtime only when a certificate is requested, while leaving standard MoE inference unchanged.
%

%

Our contributions are summarized as follows:
\begin{itemize}
    \item We identify a fundamental mismatch between existing robustness certification and MoE LLMs: existing certifiers assume a fixed computation graph inside the model, whereas MoE routing turns the certified model into different input-dependent functions, resulting in arbitrary certificates that are frequently violated in practice.

    \item We propose \textsc{CertMoE}, a framework that enables robustness certification for MoE LLMs by instantiating existing certifiers with routing-aware guarantees. \textsc{CertMoE} decomposes certification into relaxed routing containment and output stability, and combines their certificates via a union bound to provide end-to-end robustness guarantees.
    

    \item Our empirical results show that \textsc{CertMoE} substantially reduces certificate violations. Further ablations analyze the trade-off induced by certified expert-pool size and demonstrate that \textsc{CertMoE} enables failure attribution with moderate on-demand certification overhead.
\end{itemize}


\section{Background and Related Work}
\label{sec:prelim}
\noindent \textbf{Mixture-of-Experts LLMs.}
MoE architectures have been widely adopted to scale large language models efficiently, as demonstrated by systems such as Switch Transformer~\cite{fedus2022switch} and Deepseek-V2~\cite{deepseekv2}. 
In Transformer-based MoE models, the standard feed-forward network (FFN) in a transformer block is typically replaced by an MoE layer, which consists of a router and multiple experts, where each expert is itself an independent FFN, as illustrated in~\F~\ref{fig:motivate_arch}.  Formally, at layer $l$, let $\mathcal{E} = \{E_1, E_2, \dots, E_N\}$ denote 
the pool of $N$ experts. Given an input sequence $\mathbf{x}=\{x_1, x_2, 
\dots, x_n\}$, let $h_t \in \mathbb{R}^D$ denote the hidden representation 
of the $t$-th token. The router computes a gating distribution 
over experts as $\mathbf{p}_t = \operatorname{Softmax}(W_r h_t),$
where $W_r \in \mathbb{R}^{N\times D}$ are the learnable routing parameters. 
Finally, the 
active experts are selected as $\mathcal{K}_t = \operatorname{Top-K}(\mathbf{p}_t)$, 
i.e., the top-$K$ experts ranked by $\mathbf{p}_t$~\cite{jiang2024mixtral}. Overall, an MoE layer performs conditional computation: different inputs may activate different expert subsets and therefore induce different computation paths.


\noindent \textbf{Robustness Guarantee.}
Unlike heuristic defenses that can be broken by adaptive attacks~\cite{croce2020reliable}, certified robustness provides provable guarantees against adversarial perturbations. Following existing work on text robustness~\cite{cohen2019certified}, we consider token-level perturbations defined by Hamming distance.
Formally, consider an input sequence $\mathbf{x} \in \mathcal{X}$ and a text classifier $f: \mathcal{X} \rightarrow \mathcal{Y}$. Let $f(\mathbf{x})=y$ and  $\mathcal{B}(\mathbf{x}, d)$ represents a perturbation ball centered at $\mathbf{x}$ with radius $d\in\mathbb{N}$:
\begin{equation}
    \mathcal{B}(\mathbf{x}, d) = \{ \mathbf{x}' \in \mathcal{X} \mid ||\mathbf{x}-\mathbf{x}'||_0 \leq d \}.
\end{equation}
where $||\mathbf{x}-\mathbf{x}'||_0=\sum_{i=1}^n\mathbb{I}\{x_i\neq x_i'\}$ is the Hamming distance and $\mathbb{I}(\cdot)$ is the indicator function. The classifier $f$ is considered \textit{certifiably robust} at input $\mathbf{x}$ within radius $d$ if $: 
    \forall \mathbf{x}' \in \mathcal{B}(\mathbf{x}, d), \quad f(\mathbf{x}')=y. $
The certified radius $d$ quantifies the model's tolerance to perturbations. Generally, the certification objective is to design a method that, given an arbitrary input $\mathbf{x}$, certifies a radius $d$ within which the prediction is guaranteed to remain correct.

\noindent \textbf{Randomized Smoothing on Text.}
Robustness certification has been widely studied for neural networks, including
exact verification and relaxation-based methods~\cite{singh2019abstract,mao2023understanding,
muller2022prima,gowal2018effectiveness}. However, these methods often face
scalability challenges on modern deep models. Recently, randomized smoothing has emerged
as a scalable alternative by certifying a smoothed classifier through repeated
noisy inference~\cite{cohen2019certified,zeng2023certified,
zhang2023diffsmooth}. In the text domain, this is commonly instantiated by
randomly masking or perturbing tokens and aggregating predictions over the
resulting noisy inputs~\cite{zeng2023certified,ji2024advancing}. We adopt this
paradigm because it is compatible with discrete token perturbations and
applicable to LLM-scale models.

Consider a text classification task with input $\mathbf{x}=\{x_1, \dots, 
x_n\}$ and label $y \in \mathcal{Y}$. Let $f: \mathcal{X} \rightarrow 
\mathcal{Y}$ denote a base LLM classifier. Consistent with prior work~\cite{zeng2023certified},
we adopt a masking-based randomized smoothing strategy: given a 
perturbation ratio $m \in [0,1]$, we randomly select $s=\lfloor(1-m)\cdot 
n\rfloor$ positions of $\mathbf{x}$ to keep, and then replace the remaining 
tokens with the special \texttt{[MASK]} token. Let $\mathcal{M}(\mathbf{x}, 
\mathcal{H})$ denote this masking operation, where $\mathcal{H} \subseteq 
\{1,\dots,n\}$ is the set of kept (unmasked) indices, e.g., for $\mathbf{x}=$
``A B C D E'' and $\mathcal{H}=\{1,3,5\}$, $\mathcal{M}(\mathbf{x}, 
\mathcal{H})=$``A [MASK] C [MASK] E''. The smoothed classifier is 
defined as
\begin{equation}
g(\mathbf{x}) = \arg\max_{c \in \mathcal{Y}} 
\mathbb{P}_{\mathcal{H} \sim \mathcal{U}_{n,s}} \big( f(\mathcal{M}(\mathbf{x}, \mathcal{H})) = c \big),
\label{eq:base_smooth}
\end{equation}
where $\mathcal{U}_{n,s}$ denotes the uniform distribution over all 
$\binom{n}{s}$ possible selections of $s$ positions from $n$ positions, 
and predictions are aggregated over multiple sampled masks. This 
allows us to derive a certified radius $d$ for the prediction $g(\mathbf{x})$, 
ensuring that the model prediction remains invariant.
\begin{figure}[t]
  \centering
  \includegraphics[scale=0.3]{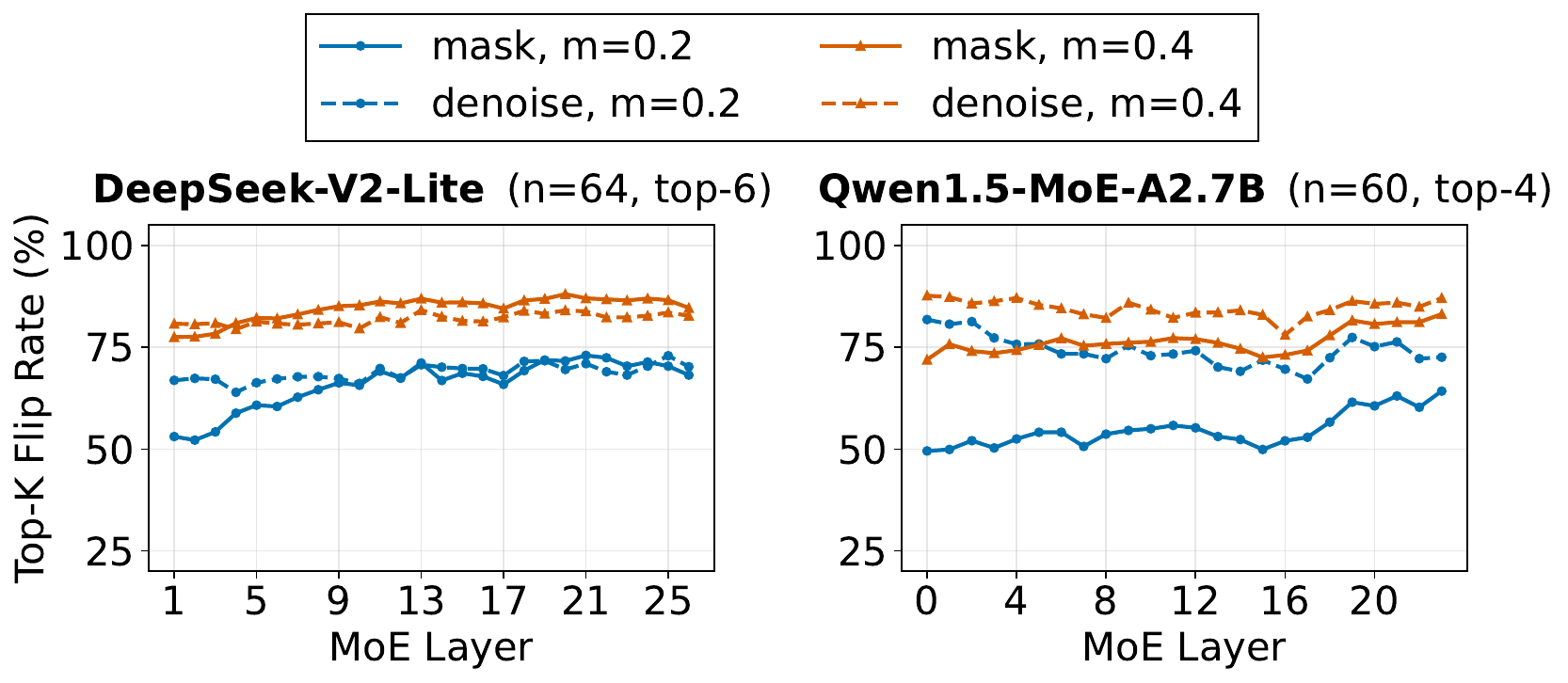}
\caption{Layer-wise Top-$K$ routing flip rates on SST-2, measured as the percentage of tokens whose selected Top-$K$ expert set differs from that of the clean input at each MoE layer.}

  \label{fig:topk_flip_ratio}
\end{figure}

\section{Motivation and Application Scope}
\label{sec:motivation}

\noindent \textbf{Key Insight.} 
The fundamental limitation of applying existing certification techniques to MoE models lies in the \textit{routing-induced dynamics} of the computation graph. Existing certification methods, such as randomized smoothing~\cite{cohen2019certified}, implicitly assume the underlying model is a fixed, deterministic function of the input. However, in MoE LLMs, the Top-$K$ routing is inherently discrete and highly sensitive to input variations. As a result, perturbing the input does not merely perturb the output; it can also change which function (i.e., which subset of experts) is being certified.

\F~\ref{fig:topk_flip_ratio} shows a layer-wise analysis on DeepSeek-V2-Lite and Qwen1.5-MoE-A2.7B, where a majority of tokens (over 50\%) undergo Top-$K$ routing flips due to input perturbations. 
confirming that input perturbations can substantially alter expert selection and introducing discontinuities in the predictive function. Consequently, the randomness involved in the certification process reflects not only output uncertainty but also routing instability. This leads to \emph{unreliable certified radii} that do not faithfully characterize the true robustness of the model.
Based on this analysis, we point out that robustness guarantees in MoE models should not be conducted in an end-to-end way, and that meaningful certification instead requires a decoupled hierarchical approach.


\noindent \textbf{Positioning \& Application Scope.}
Following the standard robustness certification setting~\cite{cohen2019certified,Gehr2018AI2SA}, \textsc{CertMoE} is an on-demand certification tool: it is invoked only when a certificate is required for a given input, while the MoE model otherwise performs standard inference. Table~\ref{tab:positioning} highlights this positioning: \textsc{CertMoE} extends verification-style certification from fixed computation graphs to LLM-scale MoE models with input-dependent execution paths. We evaluate \textsc{CertMoE} on text-only MoE LLMs because they provide a representative and practically important setting where input-dependent routing directly affects certification. Since most text robustness certification methods focus on classification-style outputs, we formulate robustness as label invariance. Nevertheless, the framework is not tied to this setting and may be extended to other tasks and modalities, as discussed in Sec.~\ref{sec:discussion}.

\begin{table*}[t]
\centering
\caption{Positioning of \textsc{CertMoE} relative to representative AI/ML 
certification and MoE robustness analysis.
D1: Formal or probabilistic guarantee; D2: Applicable to LLM-scale 
models; D3: Handles Dynamic Computation Graphs.}
\label{tab:positioning}
\footnotesize
\setlength{\tabcolsep}{4pt}
\renewcommand{\arraystretch}{1.05}
\begin{tabular}{p{0.19\linewidth} p{0.53\linewidth} 
                >{\centering\arraybackslash}p{0.07\linewidth}
                >{\centering\arraybackslash}p{0.07\linewidth}
                >{\centering\arraybackslash}p{0.07\linewidth}}
\toprule
\textbf{Category} & \textbf{Representative Methods} &
\textbf{D1} & \textbf{D2} & \textbf{D3} \\
\midrule
IBP / Linear Relaxation &
Tight Linear Approx~\cite{10.1145/3551349.3556907} (ASE '22), SharedCert~\cite{10.1007/978-3-031-13185-1_7} (CAV '22) &
\cmark & \xmark & \xmark \\
\addlinespace[3pt]

Abstract Interpretation / Convex Abstraction &
PRIMA~\cite{10.1145/3498704} (POPL '22), Abstraction-Refinement~\cite{10.1145/3544790} (TOSEM '24), PreBoxes~\cite{10.1145/3650212.3680393} (ISSTA '24), FairQuant~\cite{kim2025fairquant} (ICSE '25) &
\cmark & \xmark & \xmark \\
\addlinespace[3pt]

Solver-based Verification &
QVIP~\cite{10.1145/3551349.3556916} (ASE '22), Fairify~\cite{10.1109/ICSE48619.2023.00134} (ICSE '23), VeriStable~\cite{10.1145/3643765} (FSE '24) &
\cmark & \xmark & \xmark \\
\addlinespace[3pt]

Randomized Smoothing &
DiffSmooth~\cite{zhang2023diffsmooth} (USENIX Security '23), RanMASK~\cite{zeng2023certified} (ACL '23), SelfDenoise~\cite{ji2024advancing} (NAACL '24) &
\cmark & \cmark & \xmark \\
\addlinespace[3pt]

MoE Robustness Analysis &
AdvMoE~\cite{zhang2023robust} (ICCV '23), JTDMoE~\cite{zhangoptimizing} (ICML '25), RobMoE~\cite{puigcerver2022adversarial} (NeurIPS '22) &
\xmark & \cmark & \cmark \\

\midrule
\textbf{Routing-aware Certification} &
\textbf{\textsc{CertMoE} (Ours)} &
\cmark & \cmark & \cmark \\
\bottomrule
\end{tabular}
\end{table*}

\section{CertMoE: Routing-Aware MoE Certification}
\label{sec:certification}

\begin{figure*}[t]
  \centering
  \includegraphics[scale=0.37]{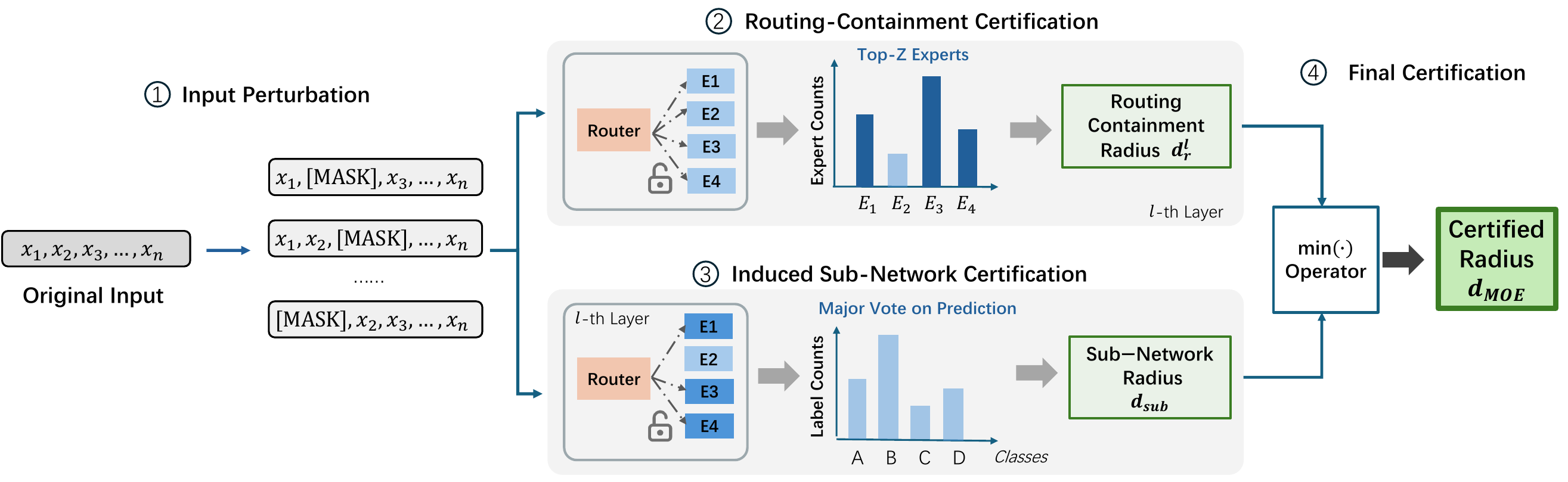}
\caption{
Overview of the \textsc{CertMoE} framework. \textsc{CertMoE} generates masked perturbations, certifies a Top-$Z$ routing-containment pool, certifies the induced sub-network restricted to this pool, and combines the two radii to obtain the final certified radius. The figure illustrates one MoE layer with $K=2$ and $Z=3$.
}
  \label{fig:method}
\end{figure*}

\noindent \textbf{Overview.}
\textsc{CertMoE} decomposes MoE robustness certification into two stages:
\emph{Routing-Containment Certification} first certifies that perturbed Top-$K$ routing remain within certified expert pools;
\emph{Induced Sub-Network Certification} then constructs the induced sub-network from these pools and applies an output-level certifier to verify prediction invariance.
This modular design separates MoE-specific routing certification from prediction certification, allowing the latter to be instantiated with existing output-level certifiers.
An overview of our method is shown in~\F~\ref{fig:method}.

\subsection{Routing Containment Certification}
\label{sec:cert_router}
This section defines smoothed routing for MoE layers and establishes a relaxed routing robustness guarantee.

\noindent \textbf{Routing Aggregation and Smoothed Routing.} 
Certifying token-level routing decisions is often impractical and overly restrictive: the certification space grows with sequence length and model depth, while routing changes for individual tokens may not affect the effective computation path or final prediction. Inspired by prior routing studies~\cite{wu2025gatebreaker, laisafex},  we aggregate token-level routing to a sentence-level routing decision. Formally, we denote the \textit{routing decision function} that maps an input sequence $\mathbf{x}$ to a set of top-$K$ activated experts as:
\begin{align}
R(\mathbf{x}) = \operatorname{Top-K}\left(
 \sum_{t=1}^{n} \mathbb{I}\!\left(e \in \operatorname{Top-K}(\mathbf{p}_t)\right) 
\right).
\end{align}
where $\mathbb{I}(\cdot)$ denotes the indicator function, $e\in\mathcal{E}$ is the selected expert, and $\mathbf{p}_t$ represents the routing probability for the $t$-th token. This counting-based representation avoids certifying each
token-level routing decision separately and provides a tractable object for
randomized smoothing.

Then, we smooth the routing decision by introducing randomness into the router's input. Given an input sequence $\mathbf{x}$ of length $n$, we randomly keep $s_r=\lfloor(1-m_r)\cdot n\rfloor$ tokens and mask the remaining positions. We denote the resulting masked input by $\mathcal{M}(\mathbf{x},\mathcal{H}_r)$, where $\mathcal{H}_r\sim\mathcal{U}_{n,s_r}$ is sampled uniformly from all subsets of $s_r$ token indices. 
For each expert $e\in\mathcal{E}$, we define its \textit{smoothed routing probability} as
\begin{align}
q_e(\mathbf{x}) =
\mathbb{P}_{\mathcal{H}_r\sim\mathcal{U}_{n,s_r}}
\left(e\in R(\mathcal{M}(\mathbf{x},\mathcal{H}_r))\right),
\label{eq:smoothed_prob}
\end{align}
where $R(\cdot)$ denotes the Top-$K$ expert selection. The \textit{smoothed routing decision} is to select the top-$K$ experts according to these probabilities:
\begin{align}
\label{eq:smooth_routing_dec}
\mathcal{T}_K(\mathbf{x}) := \operatorname{Top-K} \, (q_e(\mathbf{x})).
\end{align} 
Intuitively, $q_e(\mathbf{x})$ measures how likely expert $e$ is to be selected when the router observes randomly masked versions of the input. This expert-level characterization avoids reasoning over the $\binom{N}{K}$ possible Top-$K$ expert subsets, making it more scalable for large MoE models. 
These probabilities provide the basis for identifying consistently selected experts and constructing the certified expert pool.

\noindent \textbf{Relaxed Routing Containment.}
A straightforward notion of routing robustness is \emph{exact Top-$K$ stability}, 
where the selected experts remain unchanged under input perturbations. 
However, this requirement is often too restrictive for MoE LLMs. 
First, experts may provide redundant or complementary functions~\cite{huangmixture}, 
so partial changes within the Top-$K$ set may still lead to comparable model behavior. 
Second, token-level routing is inherently input-dependent: tokens with different 
semantic roles may activate different expert subsets, and enforcing exact 
Top-$K$ invariance can unnecessarily restrict the model's computation. 
To balance routing stability and model expressiveness, we introduce 
\textit{$\rho$-relaxed routing containment}, which relaxes exact 
Top-$K$ matching by requiring the Top-$K$ experts selected under perturbation 
to remain within a larger certified expert pool.

\begin{definition}[$\rho$-Relaxed: Certified Routing Containment]
Let $\mathcal{T}_K(\mathbf{x})$ denote the top-$K$ selected experts for 
input $\mathbf{x}$. For a given $\rho \in (0,1)$, define $Z = \lceil \rho N \rceil$, 
where $N$ is the total number of experts and $Z \ge K$. 
We say that the routing at $\mathbf{x}$ satisfies $d_r$-robust 
$\rho$-relaxed containment if, for any 
$\mathbf{x}'$ satisfying $\|\mathbf{x}' - \mathbf{x}\|_0 \le d_r$, we have
\begin{align}
\mathcal{T}_K(\mathbf{x}') \subseteq \mathcal{T}_Z(\mathbf{x}).
\label{eq:robust_router}
\end{align}
Here, $\mathcal{T}_Z(\mathbf{x})=\operatorname{Top-Z}(q_e(\mathbf{x}))$ 
denotes the certified expert pool formed by selecting the $Z$ experts with 
the largest smoothed routing probabilities.
\end{definition}

The parameter $\rho$ is a certification-time hyperparameter,  similar to the smoothing parameters used in randomized smoothing, such as the masking ratio or noise variance.
Such hyperparameters must be fixed before certifying test samples ~\cite{zhang2023diffsmooth} or calibration procedures~\cite{sun2026dual}. 
In practice, we select $\rho$ on a small held-out validation set, allowing it to adapt to each model and dataset without tuning on the certified test set.
Once selected, $\rho$ is fixed for all test samples in the corresponding dataset. 
The following theorem gives a sufficient condition for certifying the 
$\rho$-relaxed routing containment.

\begin{theorem}[Certification Condition for $\rho$-Relaxed Routing Containment]
   Let $q_e(\mathbf{x})$ and $q_e(\mathbf{x}')$ denote the smoothed routing probabilities of expert $e$
   for inputs $\mathbf{x}$ and $\mathbf{x}'$, respectively. Let $\mathcal{T}_K(\mathbf{x})$ and $\mathcal{T}_Z(\mathbf{x})$ be the set of top-$K$ and top-$Z$ expert sets at $\mathbf{x}$, respectively. For text $\mathbf{x}$, $\mathbf{x}'$, $\|\mathbf{x}' - \mathbf{x}\|_0 \le d_{\mathrm{r}}$, 
   if the following condition holds:
   \begin{align}
    \min_{e \in \mathcal{T}_K(\mathbf{x})} \underline{q}_e(\mathbf{x}) - \max_{j \notin \mathcal{T}_Z(\mathbf{x})} \overline{q}_j(\mathbf{x}) > 2 \beta \, \Delta_r,
    \label{eq:topk-routing-cert}
   \end{align}
where
\begin{align}
    \Delta_{r} = 1 - {\binom{n-d_{\mathrm{r}}}{s_r}}/{\binom{n}{s_r}}
\end{align}
and
\begin{align}
\beta := \max \Big\{ \;
&\mathbb{P}\Big( e \in R(\mathcal{M}(\mathbf{x}, \mathcal{H}_r))
    \,\Big|\, \mathcal{H}_r \cap \mathcal{I} \neq \emptyset \Big), \notag \\
&\mathbb{P}\Big( e \in R(\mathcal{M}(\mathbf{x}', \mathcal{H}_r))
    \,\Big|\, \mathcal{H}_r \cap \mathcal{I} \neq \emptyset \Big)
\;\Big\}.
\end{align}
then, with probability at least $1-\alpha$, $\mathcal{T}_K(\mathbf{x}') \subseteq \mathcal{T}_Z(\mathbf{x})$.
\label{the:routing_robust}
   \end{theorem} Here, $\beta$ quantifies the worst-case routing probability under the event that the sampled mask intersects with perturbed positions; and $\Delta_r$ represents the probability that at least one perturbed token position is included in the randomly masking set $\mathcal{H}_r$. 
\begin{proof}
Given text $\mathbf{x}$, $\mathbf{x}'$, for any expert $e\in\mathcal{E}$,  we have:
\begin{align}
|q_e(\mathbf{x}) - q_e(\mathbf{x'})|\leq \beta\Delta_r,
\label{eq:in_eq}
\end{align}
with the proof deferred to the supplementary material (\S~\ref{sec:avai}).
We then use this result to derive the main argument.
For the input $\mathbf{x}'$, we have:
 \begin{align}
    &\forall e \in \mathcal{T}_K(\mathbf{x}), \quad q_e(\mathbf{x}') \ge \underline{q}_e(\mathbf{x}) - \beta \Delta_r, \\
    &\forall j \notin \mathcal{T}_Z(\mathbf{x}), \quad q_j(\mathbf{x}') \le \overline{q}_j(\mathbf{x}) + \beta \Delta_r.
 \end{align} 
By the condition in~\E~\ref{eq:topk-routing-cert}, it follows that
\begin{equation}
\min_{e \in \mathcal{T}_K(\mathbf{x})} q_e(\mathbf{x}') >
\max_{j \notin \mathcal{T}_Z(\mathbf{x})} q_j(\mathbf{x}').
\end{equation}
Thus, under $\mathbf{x}'$, the $K$ experts in $\mathcal{T}_K(\mathbf{x})$ still have higher routing probabilities than any expert outside $\mathcal{T}_Z(\mathbf{x})$. 
Since $\mathcal{T}_K(\mathbf{x}) \subseteq \mathcal{T}_Z(\mathbf{x})$, these $K$ experts are all inside the certified pool, and no outside expert can enter the Top-$K$ set under $\mathbf{x}'$. 
Therefore, $\mathcal{T}_K(\mathbf{x}') \subseteq \mathcal{T}_Z(\mathbf{x})$.
    
\end{proof}

\subsection{Induced Sub-Network Certification}
\label{sec:cert_sub}
Given the certified expert pools, \textsc{CertMoE} constructs the induced sub-network and applies a base smoothing certifier to verify its prediction invariance.
\begin{definition}[Induced Sub-Network under Certified Expert Pools]
Let $f$ be a MoE LLM with top-$K$ routing, and let
$\{\mathcal{T}_Z^{(l)}(\mathbf{x})\}_{l\in\mathcal{L}}$
denote the certified expert pools obtained at input $\mathbf{x}$.
The \emph{induced sub-network} $f_{\mathrm{sub}}$ is defined by restricting, 
at each MoE layer $l$, the candidate experts to the certified pool 
$\mathcal{T}_Z^{(l)}(\mathbf{x})$, while the router still recomputes routing 
scores on the given input and performs top-$K$ selection among these 
restricted candidates. For any valid perturbed input $\mathbf{x}'\in\mathcal{X}$ 
around $\mathbf{x}$,
\begin{align}
    f_{\mathrm{sub}}(\mathbf{x}')
    :=
    f\big(
    \mathbf{x}';
    \{\mathcal{T}_Z^{(l)}(\mathbf{x})\}_{l\in\mathcal{L}}
    \big).
\end{align}
Here, $\mathcal{L}=\{1,2,\ldots,L\}$ denotes the set of MoE layers.
\end{definition}

Concretely, the induced sub-network keeps the original router unchanged but
restricts its candidate expert set to the certified pool at each MoE layer. That is,
the router logits are computed normally, logits of experts outside
$T_Z^{(l)}(x)$ are masked out, and the model applies its original top-$K$ routing
rule within the remaining candidates. When the routing-containment event holds
for every layer, this restriction does not alter the original MoE forward
computation: the experts selected by the original router already lie within the
certified pools, and therefore the original MoE and the induced sub-network
produce identical outputs.

Then, we define the \emph{smoothed classification probability} of class $c\in\mathcal{Y}$ as
\begin{align}
\pi_c(\mathbf{x}) = \mathbb{P}_{\mathcal{H}_o\sim\mathcal{U}_{n,s_o}}\big(f_{\mathrm{sub}}(\mathcal{M}(\mathbf{x}, \mathcal{H}_o))=c\big),
\end{align}
where $\mathcal{H}_o$ refers to the set of indices of tokens where $s_o=\lfloor(1-m_o)\cdot n\rfloor$ 
positions of $\mathbf{x}$ are kept (not masked). The prediction of the smoothed sub-network is then
$
g_{sub}(\mathbf{x}) = \arg\max_{c \in \mathcal{Y}} \pi_c(\mathbf{x}).
$
We next establish a robustness condition guaranteeing classification stability 
under input perturbations.

\begin{theorem}
[Induced Sub-Network Robustness Condition]
Let $g_{sub}(\mathbf{x})=y$ be its predicted class.  
For any input $\mathbf{x}'$ such that $\|\mathbf{x}' - \mathbf{x}\|_0 \le d_{\mathrm{sub}}$, if
\begin{align}
\pi_y(\mathbf{x}) - \max_{j \neq y} \pi_j(\mathbf{x}) > 2\gamma\Delta_{\mathrm{sub}},
\label{eq:smoothing-cert}
\end{align}
where 
    $\Delta_{\mathrm{sub}} = 1- {\binom{n-d_{\mathrm{sub}}}{s_o}}/{\binom{n}{s_o}}$
and $\gamma$ denotes the conditional probability that the classifier predicts class $c$ when the sampled mask retains at least one perturbed word. 
Then, with probability at least $(1-\delta)$, we have $
g_{sub}(\mathbf{x}') = y$.
\end{theorem}
\begin{proof}
   The proof also follows from Eq.~\eqref{eq:in_eq}; see Section~1.2 of the supplementary material (\S~\ref{sec:avai}) for details.
\end{proof}

\begin{algorithm}[t]
\small
\caption{Routing-Aware MoE Certification (CertMoE)}
\label{alg:certmoe}
\begin{algorithmic}[1]
\Require $\mathbf{x}$, $y^*$, MoE $f$ with routers $\{R^{(l)}\}_{l\in\mathcal{L}}$,
         $s_r$, $s_o$, $N_r$, $N_r'$, $N_o$, $N_o'$, $K$, $Z$, $\alpha$, $\delta$
\Ensure Certified radius $d_{\mathrm{MoE}}$

\Function{CertifyRout}{$\mathbf{x}, R^{(l)}, s_r, N_r, N_r', K, Z, \alpha$}
    \State Sample $\{\mathcal{H}_i\}_{i=1}^{N_r}\sim\mathcal{U}(n,s_r)$ and compute routing counts $C_r[e]$
    \State $\mathcal{T}_K^{(l)} \gets \operatorname{Top\text{-}K}(C_r)$,\enskip
           $\mathcal{T}_Z^{(l)} \gets \operatorname{Top\text{-}Z}(C_r)$
    \State Sample $\{\mathcal{H}'_i\}_{i=1}^{N_r'}\sim\mathcal{U}(n,s_r)$ and compute $C'_r[e]$
    \State Derive $\underline{q}_e$ ($e\in\mathcal{T}_K^{(l)}$), $\overline{q}_j$ ($j\notin\mathcal{T}_Z^{(l)}$)
           from $C'_r$ at confidence $1-\alpha$
    \State \Return $(d_r^{(l)},\,\mathcal{T}_Z^{(l)})$ where
           $d_r^{(l)} = \max\{d : \text{Eq.}~\ref{eq:topk-routing-cert}\}$
\EndFunction

\Function{CertifySub}{$\mathbf{x}, y^*, f_{\mathrm{sub}}, s_o, N_o, N_o', \delta$}
    \State Sample $\{\mathcal{H}_i\}_{i=1}^{N_o}\sim\mathcal{U}(n,s_o)$ and compute
           $\hat{y} = \arg\max_{c}\sum_{i}\mathbb{I}[f_{\mathrm{sub}}(\mathcal{M}(\mathbf{x},\mathcal{H}_i))=c]$
    \State \Return N/A \textbf{ if } $\hat{y} \neq y^*$ \Comment{Abstain if prediction is incorrect}
    \State Sample $\{\mathcal{H}'_i\}_{i=1}^{N_o'}\sim\mathcal{U}(n,s_o)$ and compute class counts $C_o[c]$
    \State Derive $\underline{\pi}_y$, $\overline{\pi}_j$ from $C_o$ at confidence $1-\delta$
    \State \Return $\max\{d : \text{Eq.}~\ref{eq:smoothing-cert}\}$
\EndFunction

\Function{CertMoE}{$\mathbf{x}, y^*, f, \{R^{(l)}\}_{l\in\mathcal{L}}$}
    \State Compute $(d_r^{(l)},\,\mathcal{T}_Z^{(l)})$ for all $l\in\mathcal{L}$
           via $\Call{CertifyRout}{\mathbf{x}, R^{(l)}, s_r, N_r, N_r', K, Z, \alpha}$
    \State Construct $f_{\mathrm{sub}}$ restricting layer $l$'s routing to $\mathcal{T}_Z^{(l)}(\mathbf{x})$
    \State $d_{\mathrm{sub}} \gets \Call{CertifySub}{\mathbf{x}, y^*, f_{\mathrm{sub}}, s_o, N_o, N_o', \delta}$
    \State \Return N/A \textbf{ if } $d_{\mathrm{sub}} = \mathrm{N/A}$
    \Comment{Abstain if prediction is incorrect}
    \State \Return $\min\!\bigl(\min_{l \in \mathcal{L}} d_r^{(l)},\; d_{\mathrm{sub}}\bigr)$
\EndFunction
\end{algorithmic}
\end{algorithm}

\subsection{MoE Robustness Certification}
\label{sec:cert_moe}
We combine the two certification stages to derive the final robustness guarantee for MoE LLMs.
\begin{corollary}[Probabilistic Robustness Certification for MoE Models]
\label{thm:moe_prob_robustness}
Let 
$\mathcal{E}_{\mathrm{route}}
:= \{\mathcal{T}_K^{(l)}(\mathbf{x}')
\subseteq \mathcal{T}_Z^{(l)}(\mathbf{x}), \forall l \in \mathcal{L}\}$
and
$\mathcal{E}_{\mathrm{sub}}
:= \{g_{sub}(\mathbf{x}') = y\}$.
For each layer $l \in \mathcal{L}$, let $d_r^{(l)}$ be the certified
$\rho$-relaxed routing radius with confidence level $1-\alpha$.
Let $d_{\mathrm{sub}}$ be the certified robustness radius of the induced
sub-network $f_{\mathrm{sub}}$ with confidence level $1-\delta$. Define
\begin{align}
d_{\mathrm{MoE}}
=
\min\!\left(
\min_{l\in\mathcal{L}} d_r^{(l)},\;
d_{\mathrm{sub}}
\right).
\label{eq:moe_radius}
\end{align}
Then for any $\mathbf{x}' \in \mathcal{X}$ satisfying
$\|\mathbf{x}'-\mathbf{x}\|_0 \le d_{\mathrm{MoE}}$,
\begin{align}
\mathbb{P}\!\left(
\mathcal{E}_{\mathrm{route}}
\cap
\mathcal{E}_{\mathrm{sub}}
\right)
\ge
1 - L\alpha - \delta .
\label{eq:moe_prob_bound}
\end{align}
Consequently, the activated experts of the original MoE model are contained in the certified expert pools. Since
the induced sub-network agrees with the original MoE model whenever the
routing-containment condition holds, we have
$f_{\mathrm{MoE}}(\mathbf{x}')=g_{sub}(\mathbf{x}')=y$.
\end{corollary}
\begin{proof}
By the definition of $d_{\mathrm{MoE}}$, the layer-wise $\rho$-relaxed
routing-containment certificates and the sub-network certificate all apply within this
radius. 
A union bound over the $L$ routing certification failures and the sub-network
certification failure gives failure probability at most $L\alpha+\delta$. 
Under routing containment, the original MoE execution is covered by
$f_{\mathrm{sub}}$; hence sub-network invariance implies prediction invariance
of the original MoE model.
\end{proof}

\noindent \textbf{Remarks.} 
CertMoE does not require exact Top-$K$ invariance; perturbed inputs may activate different experts as long as their Top-$K$ experts remain contained in the certified expert pool. 
Routing variation within this pool is handled by the induced sub-network certification: if the prediction remains invariant within the induced sub-network, CertMoE certifies the input; otherwise, it abstains to preserve soundness.

\noindent \textbf{Practical Algorithm.}
The complete \textsc{CertMoE} procedure is summarized in~\Alg~\ref{alg:certmoe}. 
We implement both routing and sub-network certification using Monte Carlo sampling and Clopper-Pearson confidence bounds~\cite{clopper1934use} for the smoothed routing and class probabilities. 
The resulting bounds are used to verify the robustness conditions and compute the final certified radius. 
Following prior randomized smoothing work~\cite{zeng2023certified}, the constants $\beta$ and $\gamma$ are estimated from the corresponding Monte Carlo probabilities, with details deferred to the supplementary material (\S~\ref{sec:avai}).

\section{Evaluation}
\label{sec:eval}

Our evaluation answers fours research questions (RQs):

\begin{itemize}
    \item \textbf{RQ1:} Can CertMoE enable existing certification methods to provide more reliable certificates for MoE LLMs?
    \item \textbf{RQ2:} How sensitive is CertMoE to certification-time parameters such as expert-pool size and masking ratio?
    \item \textbf{RQ3:} What does CertMoE reveal about the sources of certification failures in MoE LLMs?
    \item \textbf{RQ4:} What is the computational overhead of CertMoE?
\end{itemize}



\subsection{Experimental Setup}
\noindent \textbf{Dataset and models.} Following previous works on text robustness certification~\cite{ji2024advancing,wang2026clucert}, we evaluate our method on SST-2~\cite{socher2013recursive} and AGNews~\cite{zhang2015character}, two widely adopted benchmarks for certifying text classifiers. SST-2 is a binary sentiment classification dataset, while AGNews is a four-class topic classification dataset.
For target models, we consider three popular MoE-based LLMs: \textit{Qwen1.5-MoE-A2.7B}~\cite{qwen_moe}, \textit{DeepSeek-v2-Lite}~\cite{deepseekv2} and \textit{Phi-3.5-MoE-instruct}~\cite{abdin2024phi}. As shown in~\T~\ref{tab:moe_comparison}, these models span different depths, top-K settings, and expert capacities, enabling a comprehensive evaluation across heterogeneous MoE architectures.


\noindent \textbf{Baselines.}
We focus on text-domain smoothing methods whose perturbation space is defined over discrete tokens. 
Vision-domain smoothing methods~\cite{zhang2023diffsmooth}, are not included because their theoretical guarantees are developed for continuous visual perturbation spaces.
Therefore, we compare against two representative open-source methods in this category:
(1) \textit{RanMASK}~\cite{zeng2023certified} is a standard text-domain randomized smoothing method that generates masked copies by randomly replacing words with [MASK] tokens and constructs a smoothed classifier via majority voting. (2) \textit{SelfDenoise}~\cite{ji2024advancing} leverages the
multitasking nature of LLMs to first denoise
the noisy inputs and then to make predictions
based on these denoised versions.

\begin{table}[t]
\centering
\caption{Representative MoE LLMs. ``\#Act./Total'' denotes activated/total parameters. For DeepSeek-V2-Lite, $64+2$ denotes 64 routed experts plus 2 shared experts.}
\label{tab:moe_comparison}
\scalebox{0.85}{
\begin{tabular}{lrrcr}
\toprule
\textbf{Model} & \textbf{\#Layers (L)} & \textbf{Top-K} & \textbf{\#Experts (N)} & \textbf{\#Act./ Total} \\
\midrule

\textit{Qwen1.5-MoE-A2.7B} 
& 24 & 4 & 60 & 2.7B / 14.3B \\

\textit{DeepSeek-V2-Lite} 
& 27 & 6 & 64+2 & 2.4B / 16.0B \\

\textit{Phi-3.5-MoE-Instruct}
& 32 & 2 & 16 & 6.6B / 42.0B \\

\bottomrule
\end{tabular}}
\end{table}

\begin{table*}[t]
\centering
\caption{Certified robustness on SST-2 and AGNews datasets under different MoE models and certified methods. CR denotes the certified ratio, i.e., the fraction of inputs receiving a non-zero certificate, and should not be used as a standalone indicator of robustness. VR denotes the violation rate and is reported as a percentage.}
\label{tab:main_results}
\scalebox{0.82}{
\begin{tabular}{l|l
rrrrr S[table-format=1.2e-1]
rrrrr S[table-format=1.2e-1]}
\toprule
\multirow{2}{*}{Models} & \multirow{2}{*}{Method}
& \multicolumn{6}{c}{SST2}
& \multicolumn{6}{c}{AGNews} \\
\cmidrule(lr){3-8} \cmidrule(lr){9-14}
&
& {CR}
& {VR$_{(\text{DB})}$ $\downarrow$}
& {VR$_{(\text{TF})}$ $\downarrow$}
& {VR$_{(\text{TB})}$ $\downarrow$}
& {$\widehat{VR}$ $\downarrow$}
& {p-value}
& {CR}
& {VR$_{(\text{DB})}$ $\downarrow$}
& {VR$_{(\text{TF})}$ $\downarrow$}
& {VR$_{(\text{TB})}$ $\downarrow$}
& {$\widehat{VR}$ $\downarrow$}
& {p-value} \\
\midrule

\multirow{4}{*}{\shortstack{Qwen1.5-\\MoE-A2.7B}}
& RanMASK
& 0.94 & 7.45 & 30.85 & 7.45 & 32.98 & 1.11e-09
& 0.89 & 3.37 & 8.99 & 4.49 & 10.11 & 0.537 \\

& \quad + \textsc{CertMoE}
& 0.64 & 3.12 & 6.25 & 0.00 & 6.25 & 0.894
& 0.78 & 1.28 & 5.13 & 3.85 & 6.41 & 0.901 \\

& SelfDenoise
& 0.96 & 4.17 & 37.05 & 10.42 & 38.54 & 1.25e-13
& 0.84 & 4.76 & 16.67 & 8.33 & 17.86 & 0.0189 \\

& \quad + \textsc{CertMoE}
& 0.42 & 0.00 & 9.52 & 9.52 & 11.90 & 0.412
& 0.81 & 4.94 & 12.35 & 6.17 & 17.28 & 0.0295 \\

\midrule

\multirow{4}{*}{\shortstack{DeepSeek-\\v2-Lite}}
& RanMASK
& 0.89 & 13.48 & 37.08 & 15.73 & 39.33 & 2.80e-13
& 0.72 & 17.81 & 27.39 & 23.29 & 30.13 & 1.21e-06 \\

& \quad + \textsc{CertMoE}
& 0.44 & 2.27 & 9.09 & 6.82 & 9.09 & 0.653
& 0.31 & 0.00 & 6.45 & 6.45 & 6.45 & 0.830 \\

& SelfDenoise
& 0.86 & 5.43 & 32.61 & 17.39 & 39.13 & 5.08e-13
& 0.65 & 3.08 & 13.85 & 6.15 & 13.85 & 0.199 \\

& \quad + \textsc{CertMoE}
& 0.22 & 0.00 & 0.00 & 0.00 & 0.00 & 1.00
& 0.36 & 0.00 & 11.11 & 0.00 & 11.11 & 0.491 \\

\midrule

\multirow{4}{*}{\shortstack{Phi-3.5-MoE\\-Instruct}}
& RanMASK
& 0.94 & 12.77 & 31.91 & 18.09 & 36.17 & 9.80e-12
& 0.89 & 1.12 & 7.87 & 2.25 & 7.87 & 0.798 \\

& \quad + \textsc{CertMoE}
& 0.38 & 0.00 & 0.00 & 0.00 & 0.00 & 1.00
& 0.27 & 0.00 & 0.00 & 0.00 & 0.00 & 1.00 \\

& SelfDenoise
& 0.90 & 4.40 & 34.44 & 14.44 & 35.56 & 6.61e-11
& 0.89 & 6.74 & 16.85 & 6.74 & 20.22 & 0.00285 \\

& \quad + \textsc{CertMoE}
& 0.31 & 0.00 & 0.00 & 0.00 & 0.00 & 1.00
& 0.26 & 0.00 & 0.00 & 0.00 & 0.00 & 1.00 \\

\bottomrule
\end{tabular}}
\end{table*}

\noindent \textbf{Configurations.}
We conducted all experiments on NVIDIA H800 GPUs with 80GB memory.
Following prior work~\cite{ji2024advancing}, we evaluate certified robustness on a randomly selected subset of 100 test samples per dataset, due to the high computational cost of randomized smoothing-based certification. 
For fairness, all methods are evaluated on the same sampled subset, and we use the same prediction and certification sample sizes for each certification stage. For all methods, each smoothed prediction is estimated with 100 randomly perturbed samples, and each certification bound is computed with 500 samples. 
For \textsc{CertMoE}, this sampling configuration is applied to both routing certification and sub-network certification. We set the input mask ratio to 0.2 for all methods, including both routing and sub-network certification in \textsc{CertMoE}.

To implement SelfDenoise, we employ an open-source LLM Qwen3-8B~\cite{qwen3technicalreport} to improve computational effectiveness, which denoises the noisy inputs. The overall failure probability is set to 0.10, corresponding to a confidence level of 0.90; for \textsc{CertMoE}, we allocate the failure budget by setting $L\alpha=0.05$ and $\delta=0.05$. We select the certified expert ratio $\rho$ on a held-out validation set. For SST-2, $\rho$ is set to $0.2$ for Qwen and DeepSeek models and $0.4$ for Phi; for AGNews, $\rho$ is set to $0.6$ for all models. More experiments on $\rho$ are provided in ablation studies (\S~\ref{sec:ab_studies}). To assess the reliability of the certification, we evaluate the certified samples under three attacks, TextBugger (TB)~\cite{litextbugger}, DeepWordBug (DB)~\cite{gao2018black} and TextFooler (TF)~\cite{jin2020bert}, using the default settings of the TextAttack toolkit~\cite{morris2020textattack}. We explain their usage in the evaluation metrics.

\subsection{Evaluation Metrics}
Standard metrics such as average certified radius and certified accuracy~\cite{zeng2023certified} presuppose a static routing structure, under which certificates are mathematically valid. In MoE settings, however, routing can shift under perturbation, violating this assumption and rendering such metrics ill-suited. In this paper,
we therefore evaluate from two complementary perspectives:
\emph{coverage}, measured by certified ratio,
and \emph{reliability}, measured by violation rate and $p$-value.

\noindent \textbf{Certified Ratio (CR).}
Given a test set $\{\mathbf{x}_i, y_i\}_{i=1}^{N_c}$, we define the certified ratio as the fraction of inputs for which the method returns a  non-trivial certificate:
$
    \text{CR} = \frac{1}{N_c} \sum_{i=1}^{N_c} \mathbb{I}(d_i > 0),
$
where $d_i$ is the certified radius for $\mathbf{x}_i$ and is set to $0$ when certification fails. CR measures how often the certification method issues certificates over the evaluated inputs.

\noindent \textbf{Violation Rate (VR).}
Let $\bar f$ denote the smoothed prediction function used by a given
certification method, and let $y_i=\bar f(\mathbf{x}_i)$ be the certified
label of input $\mathbf{x}_i$. Let $\mathcal{C}=\{i:d_i>0\}$ denote the set
of certified inputs. VR measures the fraction of certified inputs for which
a counterexample exists within the certified region:
\begin{equation}
    \mathrm{VR} =
    \frac{1}{|\mathcal{C}|}
    \sum_{i\in\mathcal{C}}
    \mathbb{I}\!\left(
    \exists\,\mathbf{x}'\in\mathcal{B}(\mathbf{x}_i,d_i):
    \bar f(\mathbf{x}')\neq y_i
    \right).
\end{equation}
Since exhaustive search over $\mathcal{B}(\mathbf{x}_i,d_i)$ is intractable,
we approximate VR via three attack strategies
$\mathrm{atk}\in\{\mathrm{TB},\mathrm{DB},\mathrm{TF}\}$,
each generating $T$ adversarial candidates:
\begin{equation}
    \mathrm{VR}_{(\mathrm{atk})} =
    \frac{1}{|\mathcal{C}|}
    \sum_{i\in\mathcal{C}}
    \max_{j\in\{1,\dots,T\}}
    \mathbb{I}\!\left(
    \bar f\!\left(
    \mathcal{A}_{\mathrm{atk}}(\mathbf{x}_i,d_i,j)
    \right)
    \neq y_i
    \right).
\end{equation}
We aggregate the counterexamples found by all attacks and report the
attack-observed violation rate:
\begin{equation}
    \widehat{\mathrm{VR}} =
    \frac{1}{|\mathcal{C}|}
    \sum_{i\in\mathcal{C}}
    \max_{\mathrm{atk}}
    \max_{j\in\{1,\dots,T\}}
    \mathbb{I}\!\left(
    \bar f\!\left(
    \mathcal{A}_{\mathrm{atk}}(\mathbf{x}_i,d_i,j)
    \right)
    \neq y_i
    \right).
\end{equation}
Since these attacks only explore a subset of the certified region,
$\widehat{\mathrm{VR}}$ should be interpreted as the empirical violation
rate observed by our attacks. 
All attacks share the same budget of $d_i$ word modifications, matching the
$\ell_0$ metric used in certification; for character-level attacks
(DB, TB), each perturbed word counts as one unit regardless of the number of
character-level edits.

\noindent \textbf{Statistical Significance (p-value).}
To complement VR, we employ a one-sided exact binomial test
under the null hypothesis $H_0: p \leq 0.1$, where $p$ denotes
the true violation probability among certified inputs.
The test assesses whether the observed $\widehat{VR}$
provides statistically significant evidence that the violation
rate exceeds the 10\% threshold.
A larger p-value indicates insufficient evidence to reject
this hypothesis.

\subsection{Certification Reliability (RQ1)}
\label{sec:main_results}
\noindent \textbf{Overall Certification Reliability.} We report the certified robustness results on three MoE models in~\T~\ref{tab:main_results}. Across all three models, standard certification methods, including RanMASK and SelfDenoise, often achieve relatively high certification rates (CR), but their certificates suffer from substantial violation rates (VR), especially on SST-2 where VR frequently exceeds 30\%. This indicates that the high CRs of baseline methods overstate their effective certification reliability on MoE LLMs, as many issued certificates can still be violated in practice. The extremely low $p$-values of these baselines further suggest that the observed violations are unlikely to be caused by random fluctuation, but instead reflect a systematic mismatch between existing certification and dynamic MoE routing. In contrast, augmenting these certification methods with \textsc{CertMoE} leads to a more reliability-aware certification outcome. Rather than issuing potentially unreliable certificates, \textsc{CertMoE} provides a certificate only when its hierarchical certification conditions can be established, leading to fewer CR but empirically more reliable certificates.

We further observe that most methods exhibit lower VRs on AGNews than on SST-2. One possible reason is that AGNews contains four relatively distinct topic categories, making its labels harder to change within a small certified radius, whereas SST-2 sentiment labels can be affected by subtle word substitutions. Besides, a few challenging cases where residual violations remain statistically significant after applying \textsc{CertMoE}, especially for SelfDenoise on some model--dataset pairs. This may be because SelfDenoise reconstructs masked inputs before classification and can introduce new words or local expression changes, adding extra variation beyond routing dynamics. These cases suggest that routing containment substantially mitigates routing-induced unreliability, while denoising-based certifiers require additional treatment.

\noindent \textbf{Violation Rate across Perturbation Strengths.}
We further examine empirical violation rates under varying perturbation strengths, measured as the fraction of perturbed tokens relative to sentence length, on representative SST-2 settings. As shown in Fig.~\ref{fig:vio_rate}, baseline certification methods, both RanMASK and SelfDenoise,  exhibit consistently high violation rates across different perturbation budgets, whereas incorporating \textsc{CertMoE} substantially reduces violations on both DeepSeek-V2-Lite and Qwen1.5-MoE-A2.7B. This indicates that \textsc{CertMoE} maintains low violation rates not only in the overall results but also at each individual perturbation strength, suggesting stable empirical reliability under different attack budgets. 

\begin{figure}[t]
  \centering
  \includegraphics[scale=0.25]{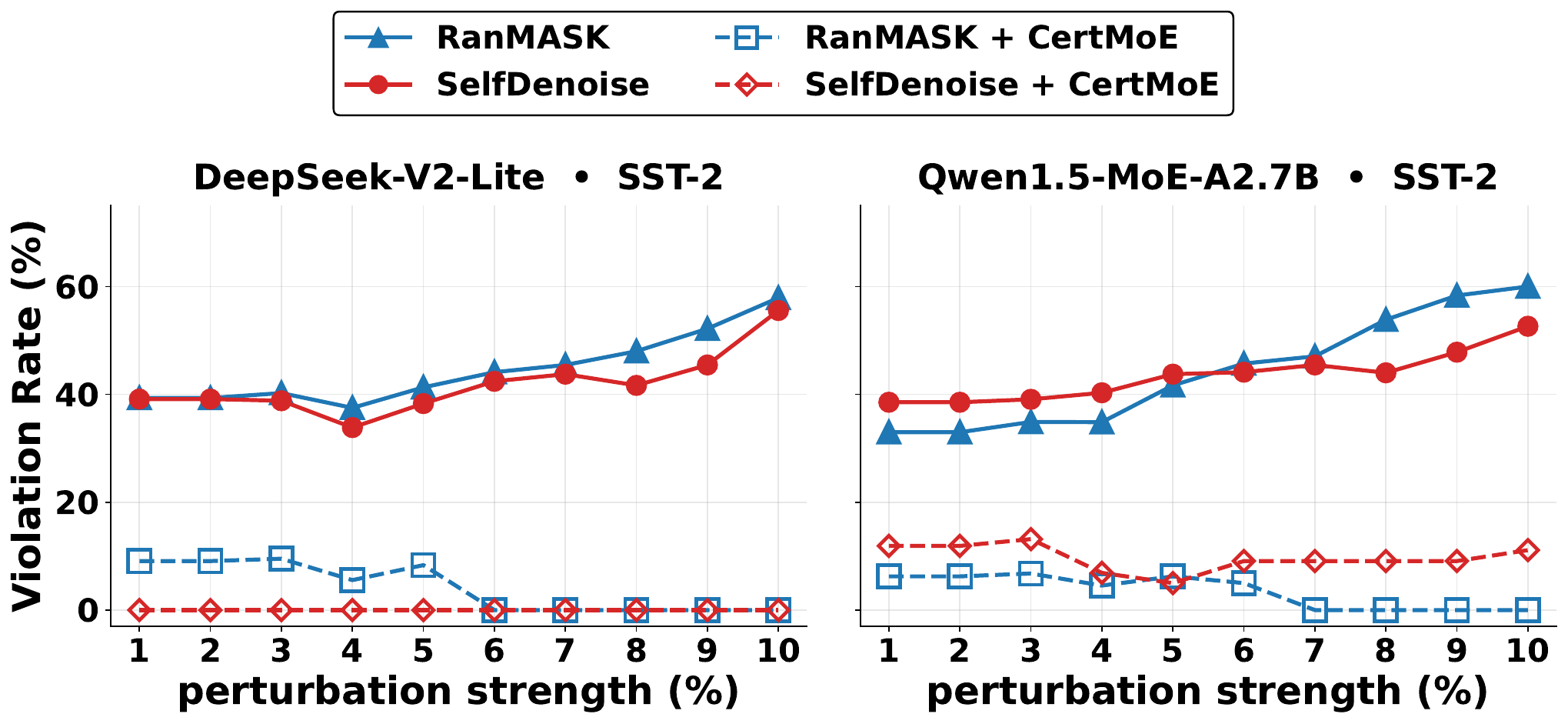}
  \caption{Violation Rate of different certification methods, plotted as a function of perturbation intensity.}
  \label{fig:vio_rate}
\end{figure}

\subsection{Parameter Sensitivity (RQ2)}
\label{sec:ab_studies}
\noindent \textbf{Impact of Certified Expert-Pool Ratio.}
We study the impact of $\rho$, which determines the size $Z=\lceil \rho N\rceil$ of the certified expert pool used for routing containment. ~\T~\ref{tab:rho_ablation} reports a sensitivity analysis over
different $\rho$ values on SST-2 with Qwen1.5-MoE-A2.7B under RanMASK. As shown, the strict setting with $Z=K$ in the first row certifies no inputs, showing that exact routing containment is too restrictive. Increasing $\rho$ enlarges the certified expert pool and makes the routing-containment condition easier to satisfy, so \textsc{CertMoE} assigns certificates to more inputs. However, this does not imply better reliability. Larger pools exhibit much higher violation rates than $\rho=0.2$. This is because an overly permissive pool weakens the constraint on perturbation-induced routing changes and requires the induced sub-network to account for a broader set of possible expert selections. These results indicate that $\rho$ controls a trade-off between the permissiveness of routing-containment certification and the reliability of the resulting certificates.
\begin{table}[t]
\centering
\caption{Impact of certified expert ratio $\rho$ on certified rate (CR) and violation rate (VR).}
\label{tab:rho_ablation}
\scalebox{0.9}{
\begin{tabular}{l r r r r rr}
\toprule
$\rho$ & \#CP & CR & VR$_{(\text{DB})}$ & VR$_{(\text{TF})}$ & VR$_{(\text{TB})}$ & $\widehat{VR}$\\
\midrule
$N/K$ & 6  & 0    & N/A   & N/A   & N/A   & N/A\\
0.2   & 12 & 0.64 & 3.12  & 6.25  & 0  & 6.25 \\
0.4   & 24 & 0.69 & 23.19 & 39.13 & 26.08 & 44.93 \\
0.6   & 36 & 0.84 & 12.05 & 43.28 & 20.48 & 48.19 \\
0.8   & 48 & 0.93 & 15.05 & 36.56 & 15.05 & 40.86\\
\bottomrule
\end{tabular}}
\end{table}

\noindent \textbf{Impact of Mask Ratio on Smoothing.}
To examine whether the reliability improvement of \textsc{CertMoE} depends on a specific mask-ratio choice, we conduct an ablation study on Qwen1.5-MoE-A2.7B using SST-2 by varying the output mask ratio $m_o$ and routing mask ratio $m_r$. We compare RanMASK with its \textsc{CertMoE}-augmented counterpart in terms of certification coverage and empirical violation ratio. The results in~\T~\ref{tab:ab_mask} show a trend consistent with the main results in~\T~\ref{tab:main_results}. Changing the mask ratio alone does not resolve the reliability issue of existing smoothing methods: RanMASK tends to certify a large fraction of inputs under different $m_o$ settings (i.e., 0.2 and 0.4), but the resulting certificates still exhibit substantial violation rates, especially under TextFooler. In contrast, adding \textsc{CertMoE} leads to more selective certification outcomes across different $(m_o, m_r)$ settings and substantially reduces VRs in practics. These results further confirm that \textsc{CertMoE} remains effective across different mask-ratio configurations, and proper mask-ratio adjustment may further improve the certificates' reliability.

\begin{table}[t]
\centering
\caption{Impact of mask settings on the robustness of RanMASK and ours across against three attacks.}
\label{tab:ab_mask}
\scalebox{0.9}{
\begin{tabular}{lrrrrrr} 
\toprule
Method & $m_o$ & $m_r$ & CR & VR$_{(\text{DB})}$ & VR$_{(\text{TF})}$ & VR$_{(\text{TB})}$ \\
\midrule
RanMASK & 0.2 & N/A & 0.94 & 7.45 & 30.85 & 7.45 \\ \midrule
\multirow{3}{*}{+ \textsc{CertMoE}} & \multirow{3}{*}{0.2} & 0.2 & 0.64 & 3.12 & 6.25 & 0.00 \\

 &  & 0.4 & 0.64 & 4.69 & 6.25 & 10.93 \\
 &  & 0.6 & 0.61 & 0.00 & 3.27 & 11.47 \\
\midrule
RanMASK & 0.4 & N/A & 0.91 & 9.89 & 34.07 & 9.89 \\
\midrule
\multirow{3}{*}{+ \textsc{CertMoE}} & \multirow{3}{*}{0.4} & 0.2 & 0.61 & 0.00 & 3.27 & 6.67 \\

 &  & 0.4 & 0.61 & 1.63 & 1.63 & 1.63 \\
 &  & 0.6 & 0.61 & 0.00 & 0.00 & 0.00 \\
\bottomrule
\end{tabular}}
\end{table}

\begin{figure}[t]
  \centering
  \includegraphics[scale=0.5]{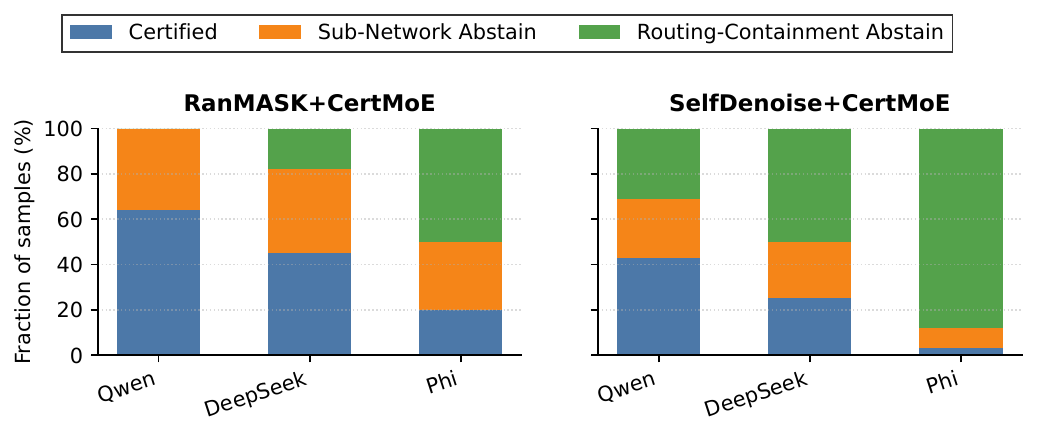}
  \caption{Certification outcome decomposition of CertMoE outcomes on SST-2 under RanMASK and SelfDenoise across MoE models. Model names are abbreviated by their prefixes.}
  \label{fig:abstain_decomp}
\end{figure}

\subsection{Failure Attribution (RQ3)}
\noindent \textbf{Abstention Analysis.}
We analyze the sources of abstention on SST-2 across three MoE LLMs using RanMASK and SelfDenoise with $\rho=0.2$. Each sample is assigned to one of three mutually exclusive outcomes, i.e., certified, sub-network abstain, or routing-containment abstain, where routing containment is checked first and only passing samples proceed to sub-network certification. The results in~\F~\ref{fig:abstain_decomp} reveal two consistent patterns. First, routing-containment abstention rises monotonically with model size, reflecting that larger MoE models operate over a more fine-grained expert space and exhibit greater structural variability in routing decisions. Second, this effect is amplified under SelfDenoise, where input rewrites alter surface form and syntax, shifting hidden representations enough to produce inconsistent expert selections across perturbation samples. In both cases, \textsc{CertMoE} withholds certification in response to routing variability, reflecting its intended conservative behavior. These results underscore that routing certification is a necessary and earlier checkpoint than prediction certification in MoE models.

\begin{table*}[t]
\centering
\setlength{\tabcolsep}{4pt}
\caption{Certification case study on SST-2. 
\textbf{Bold words} indicate the substitution target in the original text and the corresponding adversarial replacement. Each adversarial example (AE) is generated within the certified robust bound.}
\label{tab:case_study_bottleneck}
\scalebox{0.92}{
\begin{tabularx}{\linewidth}{@{}l l>{\raggedright\arraybackslash}X>{\raggedright\arraybackslash}X>{\centering\arraybackslash}p{0.09\linewidth}>{\raggedright\arraybackslash}p{0.2\linewidth}@{}}
\toprule
Sample ID & Attack Method & Original Text & Adversarial Example (AE) & RanMASK & RanMASK+CertMoE (Ours) \\
\midrule
1 & TextFooler
& is this \textbf{progress}?
& is this \textbf{\textcolor{red}{breakthrough}}?
& \textcolor{red}{Violated}
& Abstain (Sub-Network) \\
\midrule
2 & DeepWordBug
& for me, this opera is \textbf{n't} a favorite, so it is a long time before the fat lady sings.
& for me, this opera is \textbf{\textcolor{red}{i't}} a favorite, so it is a long time before the fat lady sings.
& \textcolor{red}{Violated}
& Abstain (Sub-Network) \\
\midrule
3 & TextFooler
& it is one of the most \textbf{honest} films ever made about hollywood.
& it is one of the most \textbf{\textcolor{red}{franco}} films ever made about hollywood.
& \textcolor{red}{Violated}
& Abstain (Routing-Containment) \\
\midrule
4 & TextFooler
& it is like a \textbf{poem}.
& it is like a \textbf{\textcolor{red}{anthology}}.
& \textcolor{red}{Violated}
& \textcolor{teal}{Certified \& No AE found on the smoothed function} \\
\midrule
5 & TextBugger
& the movie has a soft, percolating magi\textbf{c}, a deadpan suspense.
& the movie has a soft, percolating magi\textbf{\color{red}{{\scriptsize c}}}, a deadpan suspense.
& \textcolor{red}{Violated}
& \textcolor{teal}{Certified \& No AE found on the smoothed function} \\
\bottomrule
\end{tabularx}}
\end{table*}

\noindent \textbf{Case Studies.}
Table~\ref{tab:case_study_bottleneck} provides a fine-grained analysis of how \textsc{CertMoE} refines certification outcomes compared with baseline methods on Qwen1.5-MoE-A2.7B.
As shown, the baseline method, RanMASK, assigns certificates to all five examples.
However, each baseline certificate can be falsified by a successful adversarial substitution generated within its certified radius, indicating that these certificates may overestimate the true robustness in practice. By contrast, \textsc{CertMoE} raises abstain when certifying the first three examples, and further attributes each uncertified outcome to routing-containment or induced sub-network.
Besides, for the last two examples, \textsc{CertMoE} preserves certification when both routing and sub-network conditions are genuinely satisfied, and no adversarial substitutions can falsify its certified predictions.
Overall, these cases demonstrate that \textsc{CertMoE} yields both reliable and interpretable certificates: it issues certificates only when jointly verified by both routing and sub-network components, while component-level attribution further points to actionable directions for improving model robustness.

\subsection{Computation Cost (RQ4)}

\noindent \textbf{Runtime Cost.} 
We further analyze the per-sample computation cost of \textsc{CertMoE}. The runtime is measured on SST-2 and AGNews under three MoE LLMs. For each model, we show two bars: the left bar reports the base prediction and certification costs of RanMASK alone, and the right bar shows the total cost of \textsc{CertMoE} augmented on top of RanMASK, further decomposed into the routing-containment (RC) and sub-network (SN) prediction (Pred) and certification (Cert) costs. As shown in Fig.~\ref{fig:cost_breakdown}, \textsc{CertMoE} roughly doubles the total runtime compared to RanMASK alone across most settings. Specifically, RC certification and SN certification together constitute the primary cost, with their relative proportions varying across models and datasets. In most settings, RC Cert accounts for the larger share; however, on AGNews with Qwen, SN Cert becomes the dominant cost, as more samples pass the RC stage and proceed to sub-network certification. In contrast, Phi exhibits relatively low SN cost because a larger fraction of samples are rejected at the RC stage and thus never enter sub-network certification, highlighting the role of RC certification as an effective pre-filter that controls the overall computational overhead. We note that this overhead is a one-time cost incurred only when a certificate is requested for a certain input, and does not affect standard inference.
\begin{figure}[t]
  \centering
  \includegraphics[scale=0.55]{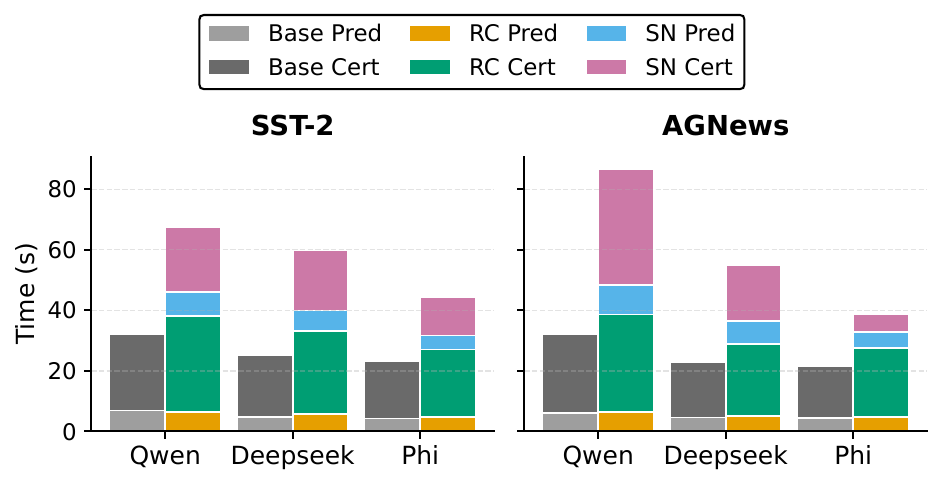}
  \caption{Per-sample computation cost breakdown w.r.t. LLMs.}
  \label{fig:cost_breakdown}
\end{figure}

\noindent \textbf{Memory Overhead.}
\textsc{CertMoE} uses the same underlying MoE model as the randomized smoothing
baseline and does not instantiate additional model copies or expert networks.
Beyond the standard memory required for model parameters and forward-pass
activations, CertMoE only maintains lightweight certification statistics. For a
model with $L$ MoE layers, $N$ experts per layer, and certified expert-pool size
$Z$, the persistent additional memory consists mainly of
expert count vectors and certified expert pools, with complexity
$O(LN + LZ)$. Since $Z \leq N$, this reduces to $O(LN)$, which is negligible
compared with the memory footprint of model parameters and activations. In our
implementation, these routing statistics are stored as small CPU-side arrays or
dictionaries, and Monte Carlo samples are processed without caching their
intermediate activations. Therefore, CertMoE does not materially increase peak
GPU memory beyond the standard batched MoE forward pass.

\section{Discussion}
\label{sec:discussion}

\noindent \textbf{Threats to Validity.}
Several threats may affect our study. Implementation errors and stochastic certifiers may influence the results; we mitigate this by using the same perturbation space, base certifiers, and evaluation protocol for baselines and \textsc{CertMoE}. Our metrics are tied to the considered perturbation model and should not be interpreted as robustness against arbitrary text attacks. The evaluation focuses on text classification and representative MoE LLMs, so generalization to other architectures or generation tasks requires further study. Finally, runtime is hardware- and implementation-dependent and is used mainly for relative cost comparison.

\noindent \textbf{Extension to Generation Tasks.} Certified robustness for sequence generation remains less developed than classification-style certification. We note that text generation can be viewed as a sequential classification task, suggesting that classification-based certification can, in principle, be applied at each generation step. Furthermore, existing certification approaches for generative models typically constrain distances in the output vector space, providing a complementary foundation for generation-time guarantees. Since \textsc{CertMoE} is a plug-in framework compatible with existing base certifiers, we view its extension to sequence generation as a promising direction for future work.

\noindent \textbf{Extension to Other Modalities.}
\textsc{CertMoE} can also apply to MoE architectures in other modalities, as long as the underlying perturbation and base certification techniques are adapted to the modality. For example, text masking can be replaced by patch- or pixel-level perturbations for vision-domain MoE models. This suggests that the main extension challenge lies in matching the certification backbone to modality-specific perturbation spaces.

\noindent \textbf{Certified Training for MoE.}
The theory of \textsc{CertMoE} can also inform future certified training for MoE models. Since certifiable robustness depends on both routing containment and prediction stability, training objectives could explicitly promote stable routing under perturbations and robust predictions within the induced sub-network. We leave the design of routing-aware certified training methods to future work.

\section{Conclusion}

This paper identifies a fundamental mismatch between existing certification methods and MoE LLMs: existing certifiers assume a fixed computation graph, whereas MoE routing makes the certified model input-dependent. To close this gap, we propose \textsc{CertMoE}, a routing-aware certification framework that explicitly accounts for routing uncertainty and its impact on robustness guarantees. Experiments on state-of-the-art MoE LLMs and benchmark datasets show that \textsc{CertMoE} consistently produces more reliable certificates.
\section{Data Availability}
\label{sec:avai}
The code, data, and supplementary material for this paper are available at \url{https://anonymous.4open.science/r/CertMoE-69CB}.

\bibliographystyle{./IEEEtran}
\bibliography{./IEEEabrv,./IEEEexample}

\end{document}